\documentclass[12pt]{article}
\usepackage{tikz}
\usepackage{amsmath,amssymb}
\usepackage{bbm}  
\usepackage{fancyhdr}
\usepackage[utf8]{inputenc}    
\usepackage{fullpage}
\usepackage{float}
\usepackage{authblk}
\usepackage{color}
\usepackage{amsthm} 

\let\OLDthebibliography\thebibliography
\renewcommand\thebibliography[1]{
  \OLDthebibliography{#1}
  \setlength{\parskip}{0pt}
  \setlength{\itemsep}{0pt plus 0.3ex}
}

\newcommand{\rvoid}{|0\rangle}

\newtheorem{theorem}{Theorem}

\newtheorem{property}{Lemma}

\usepackage{tikz}
\usetikzlibrary{patterns}
\usetikzlibrary{calc,arrows,positioning}
\usetikzlibrary{arrows,shadows}
\usetikzlibrary{decorations.pathmorphing}	
\usetikzlibrary{decorations.markings}
\author[1,2,4]{Etienne Granet}
\author[1,2,3]{Jesper Lykke Jacobsen}
\affil[1]{Institut de Physique Th\'eorique, Paris Saclay, CEA, CNRS, 91191 Gif-sur-Yvette, France}
\affil[2]{Laboratoire de Physique de l'Ecole Normale Sup\'erieure, ENS, Universit\'e PSL, CNRS, Sorbonne Universit\'e, Universit\'e Paris-Diderot, Sorbonne Paris Cit\'e, Paris, France}
\affil[3]{Sorbonne Universit\'e, \'Ecole Normale Sup\'erieure, CNRS, Laboratoire de Physique (LPENS), 75005 Paris, France}
\affil[4]{The Rudolf Peierls Centre for Theoretical Physics, Oxford
University, Oxford OX1 3PU, UK}

\title{On zero-remainder conditions in the Bethe ansatz}
\date{}

\begin{document}

\maketitle

\begin{abstract}
We prove that physical solutions to the Heisenberg spin chain Bethe ansatz equations are exactly obtained by imposing two zero-remainder conditions. This bridges the gap between different criteria, yielding an alternative proof of a recently devised algorithm based on $QQ$ relations, and solving its minimality issue.
\end{abstract}

\subsection{Introduction}
The Bethe ansatz \cite{bethe} enables one to write an energy level of the periodic Heisenberg (XXX) spin chain on $L$ sites as
\begin{equation}
E=-\frac{1}{2}\sum_{i=1}^K\frac{1}{\lambda_i^2+1/4} \,,
\end{equation}
where the $K$ Bethe roots $\lambda_i$ satisfy the Bethe ansatz equations traditionally written in the following form
\begin{equation}
\label{eq:bexxz}
\left( \frac{\lambda_i+i/2}{\lambda_i-i/2}\right)^L=\prod_{j\neq i}\frac{\lambda_i-\lambda_j+i}{\lambda_i-\lambda_j-i} \,.
\end{equation}
To make this simply stated result precise however necessitates recalling some subtleties in the resolution of \eqref{eq:bexxz}.\\

Firstly, the ordering of the $\lambda_i$'s is irrelevant, and two solutions differing only by a permutation of the roots should be considered equal; secondly, in general a solution with two coinciding roots $\lambda_i=\lambda_j$ should be discarded \cite{bethe,slv}; thirdly, there are special solutions with $\lambda_1=i/2$ and $\lambda_2=-i/2$, called \textit{exact strings} \cite{bethe,esslerfine,noh,hagemans,fabricius,bax,goetze,bazhanov}, which must be sometimes discarded and sometimes not. These issues are particularly important for e.g. completeness of the Bethe ansatz, which has been widely studied \cite{bethe,kirillov,kirillov2,klumperzittartz,esslercomplete,esslerfine,juettner,tarasov3}, but also because these exceptional solutions play a noteworthy role in a variety of contexts, see e.g. \cite{wortis,ganahl,symmetrybreaking,arutyunov}. 

It is well known that the first two points are solved by inspection of the expression of the eigenvalue $T(\lambda)$ of the transfer matrix at spectral parameter $\lambda$. It satisfies the \textit{$TQ$ relation} \cite{Baxter8Vmodel1972,BaxterBook}
\begin{equation}
\label{eq:tq}
T(\lambda)Q(\lambda)=(\lambda-i/2)^LQ(\lambda+i)+(\lambda+i/2)^LQ(\lambda-i) \,,
\end{equation}
where
\begin{equation}
Q(\lambda)=\prod_{i=1}^K(\lambda-\lambda_i) \,.
\end{equation}
From general arguments $T(\lambda)$ has to be a polynomial in $\lambda$, and thus \eqref{eq:tq} gives a zero-remainder condition on the coefficients of $Q$, that permits then to solve for the $\lambda_i$'s, and that directly removes the non-physical solutions with coinciding roots.\\

However, the $TQ$ relation alone does not solve the third point. Indeed, any solution $(i/2,-i/2,\lambda_3,...,\lambda_K)$ where $\lambda_3,...,\lambda_K$ satisfy the Bethe equations \eqref{eq:bexxz} (for $\pm i/2$, they are automatically satisfied if both sides are multiplied by the denominators beforehand) does give a polynomial $T(\lambda)$, but the normalized Bethe state is then expressed in a singular way whose regularization depends on the way the roots $\lambda_i$ of the strings converge to $\pm i/2$. More precisely, denoting the $R$-matrix by
\begin{equation}
\label{defr}
R(\lambda)=\left(\begin{matrix}
\lambda+i/2&0&0&0\\
0&\lambda-i/2&i&0\\
0&i&\lambda-i/2&0\\
0&0&0&\lambda+i/2
\end{matrix} \right)\,,
\end{equation}
where the monodromy matrix and the transfer matrix read respectively, using the standard notations of the algebraic Bethe ansatz approach \cite{slv},
\begin{equation}
\left(\begin{matrix}
A&B\\
C&D
\end{matrix} \right)=R_{0L} \cdots R_{01}\,,\qquad t=A+D
\end{equation}
we have
\begin{align}
\label{unwanted}
t(\lambda)B(\lambda_1) \cdots B(\lambda_K)\rvoid&=T(\lambda)B(\lambda_1) \cdots B(\lambda_K)\rvoid\\
&-\sum_{i=1}^K \frac{\text{Res}(T(\lambda),\lambda_i)}{\lambda-\lambda_i}B(\lambda)B(\lambda_1)...\widehat{B}(\lambda_i) \cdots B(\lambda_K)\rvoid \,,
\end{align}
where the hat $\widehat{B}$ indicates that the corresponding factor is omitted in the product. Although the residues do vanish in case of strings, we have $B(\lambda_1)...B(\lambda_K)\rvoid=0$ \cite{slv,Siddharthan}, so that the normalized Bethe state is not necessarily an eigenvector. In fact, the solutions to the Bethe equations with exact strings sometimes do yield eigenvalues and eigenvectors of the Hamiltonian, and sometimes not.\\

In \cite{nepomechie,avdeev} a sufficient condition was found under which a Bethe vector can be built which is an eigenvector, by regularizing in a very particular way the roots $\pm i/2$. Another practical way of distinguishing physical from non-physical solutions is to examine the behaviour of the solutions in terms of an additional parameter, such as a twist \cite{volin,nepomechietwist}. But the most important recent advance on this question was an efficient algorithm \cite{marboe} found by Marboe and Volin to solve the $TQ$ relation while discarding automatically the non-physical solutions among those with an exact string, with only algebraic manipulations and zero-remainder conditions called $QQ$ relations.\\

The algorithm of \cite {marboe} relies on the remarkable result of \cite {mukhin} (following previous works \cite{tarasavo2,tarasov3,pronko,bazhanov}) that the Bethe ansatz is complete provided the Bethe equations are replaced by finite difference second order equations whose two solutions are polynomial, and it was claimed by the authors of \cite{marboe} that their algorithm implies that only cases where both solutions are polynomial are obtained. In view of the long standing study of exact string solutions and their role in the completeness of the Bethe ansatz, it is natural to ask how these exceptional solutions are related to the work \cite {mukhin}, that does not deal with them at all, and whose reasoning is very far from that of  \cite{nepomechie,volin,nepomechietwist,avdeev}. In particular it does not show that the extra eigenvalues obtained by this replacement correspond to the regularization of the singular exact strings solutions to the Bethe equations -- if they did not, it would precisely mean that the Bethe ansatz equations alone are incomplete.

Another unclear point in this construction is the minimality of the algorithm of \cite {marboe}. Indeed it  relies  on a large amount (proportional to $L$) of zero remainder conditions, whereas  it was conjectured therein that imposing a much smaller number of relations would suffice to lead to the same physical solutions. However, the authors of \cite{marboe} did not provide a proof of this conjecture, suggesting that a simpler or more natural proof is lacking, and that would solve the minimality issue of their algorithm.\\

In this paper, we answer these points and show that (i) for exact strings, having two polynomial solutions to the $TQ$ relations ensures that $T(\lambda)$ is an eigenvalue of the transfer matrix, hence showing that the exact strings solutions are indeed counted in \cite {mukhin}, reconciling their approach with the regularization of  singular exact strings solutions \cite{nepomechie,volin,nepomechietwist,avdeev}; and (ii) we show that the $TQ$ relation has to be supplemented with another zero-remainder $TQ$-like relation to yield all and only physical solutions, giving an elegant way of distinguishing physical and non physical solutions, and proving the algorithm of \cite {marboe} as well as  the minimality issue raised therein. \\

These results and methods established here for the periodic XXX chain will be used in a subsequent paper \cite{otherarticle} to generalize and prove $QQ$ relations for the anisotropic Heisenberg (XXZ) spin chain, and to extend the results to the case of open boundary conditions.

\subsection{Polynomiality of the other solution to the $TQ$ relation}
We thus consider Bethe roots $\Lambda=\{\lambda_1,...,\lambda_n\}$ solutions to the equations \eqref{eq:bexxz}. We will assume that all roots are different, $\lambda_i\neq\lambda_j$ if $i\neq j$ \cite{neposommese}. We denote $\bar{\Lambda}$ the set of $\lambda_i$'s such that there does not exist another $\lambda_j$ with $\lambda_i-\lambda_j=\pm i$, and $S$ the set of complex numbers $s$ (the 'center of strings') such that $s+i/2\in\Lambda$ and $s-i/2\in\Lambda$. We denote $\bar{Q}(\lambda)=\prod_{\lambda_k\in\bar{\Lambda}}(\lambda-\lambda_k)$. We will finally use the convenient notation $Q^*(\lambda)=\prod_{k} (\lambda-\lambda_k)$ if $\lambda\notin \Lambda$ and $Q^*(\lambda_p)=\prod_{k\neq p} (\lambda_p-\lambda_k)$ for $\lambda_p\in\Lambda$. 

Before addressing the main results, for sake of completeness we recall here the following known result

\begin{property}
\label{th:0}
We have $S=\varnothing$ or $S=\{0\}$.
\end{property}
\begin{proof}
Assume that there are two roots such that $\lambda_{i_1}-\lambda_{i_2}=i$. Denote $s=\lambda_{i_1}-i/2$. Then from \eqref{eq:bexxz} with $k=i_1$, either $\lambda_{i_1}=i/2$, in which case $s=0$, or there exists another $\lambda_{i_3}$ such that $\lambda_{i_1}-\lambda_{i_3}=-i$. In the latter case, the same argument can be then repeated with $k=i_3$, so that $s=ni$ with $n$ a negative or zero integer, since there is a finite number of roots. On the other hand, \eqref{eq:bexxz} for $k=i_2$ implies that either $\lambda_{i_2}=-i/2$, in which case $s=0$, or there exists another $\lambda_{i_4}$ such that $\lambda_{i_2}-\lambda_{i_4}=i$. The same argument can be then repeated with $k=i_4$, implying that $s=ni$ with $n$ a positive or zero integer. Thus $s=0$.
\end{proof}

Let us start with the following property, that generalizes  \cite{pronko} to the exact strings case.

\begin{property}
\label{th:1}
There exist a polynomial $P_0(\lambda)$ and complex numbers $\alpha_s$ for $s\in S$ such that
\begin{equation}
\label{eq:pq}
P(\lambda+i/2)Q(\lambda-i/2)-P(\lambda-i/2)Q(\lambda+i/2)=\lambda^L
\end{equation}
with
\begin{equation}
\label{eq:pform}
P(\lambda)=P_0(\lambda)+Q(\lambda)\sum_{s\in S}\alpha_s \psi (-i(\lambda-s)+1/2) \,,
\end{equation}
where $\psi(x)$ is the digamma function. Moreover, $\alpha_0=0$ if and only if the following additional Bethe equation is satisfied
\begin{equation}
\label{eq:genb}
(-1)^L=\prod_{\lambda_k\neq \pm i/2}\frac{\lambda_k+i/2}{\lambda_k-i/2}\cdot \frac{\lambda_k+3i/2}{\lambda_k-3i/2} \,.
\end{equation}
\end{property}

%
%

\begin{proof}

It is directly inspired by \cite{pronko}, where the authors (implicitly) treated the case $S=\varnothing$.

Denote
\begin{equation}
\label{eq:r0}
R(\lambda)=\frac{\lambda^L}{Q(\lambda+i/2)Q(\lambda-i/2)} \,.
\end{equation}
We have
\begin{equation}
\frac{T(\lambda)}{Q(\lambda+i)Q(\lambda-i)}=R(\lambda+i/2)+R(\lambda-i/2) \,.
\end{equation}

Since each $s\in S$ appears twice in the denominator in $R(\lambda)$, we can decompose
\begin{equation}
\label{eq:r}
R(\lambda)=\pi(\lambda)+\frac{q_-(\lambda)}{Q(\lambda-i/2)}+\frac{q_+(\lambda)}{Q(\lambda+i/2)}+\sum_{s\in S}\frac{c_{s}}{(\lambda-s)^2}
\end{equation}
with $\pi(\lambda), q_\pm(\lambda)$ polynomials of degree less than or equal to $n-1$ (since a term of order $n$ in the numerator could be reabsorbed in the constant term of $\pi(\lambda)$), and $c_s$ complex numbers. From this one gets

\begin{equation}
\label{eq:t}
\begin{aligned}
\frac{T(\lambda)}{Q(\lambda+i)Q(\lambda-i)}= & \ \pi(\lambda-i/2)+\pi(\lambda+i/2)\\
&+\frac{q_-(\lambda-i/2)}{Q(\lambda-i)}+\frac{q_+(\lambda+i/2)}{Q(\lambda+i)}+\frac{q_+(\lambda-i/2)+q_-(\lambda+i/2)}{Q(\lambda)}\\
&+\sum_{s\in S} \left( \frac{c_{s}}{(\lambda-s+i/2)^2}+\frac{c_{s}}{(\lambda-s-i/2)^2} \right) \,.
\end{aligned}
\end{equation}
Multiplying by $(\lambda-s+i/2)^2$ and sending $\lambda\to s-i/2$, since there is no double pole in $s-i/2$ on the left-hand side, one gets $c_s=0$.

For $\lambda_j\in\bar{\Lambda}$, multiplying by $(\lambda-\lambda_j)$ and taking $\lambda\to\lambda_j$ yields
\begin{equation}
q_+(\lambda_j-i/2)+q_-(\lambda_j+i/2)=0 \,,
\end{equation}
meaning that there exists a polynomial $\sigma$ such that
\begin{equation}
q_+(\lambda-i/2)+q_-(\lambda+i/2)=\bar{Q}(\lambda)\sigma(\lambda) \,,
\end{equation}
and thus a polynomial $q(\lambda)$ such that
\begin{equation}
q_\pm(\lambda)=\pm q(\lambda\pm i/2)+\frac{1}{2}\bar{Q}(\lambda\pm i/2)\sigma(\lambda\pm i/2)
\end{equation}
(for example, take $q(\lambda)=q_+(\lambda-i/2)-\tfrac{1}{2}\bar{Q}(\lambda)\sigma(\lambda)$). Thus
\begin{equation}
\begin{aligned}
R(\lambda)= & \ \pi(\lambda)-\frac{q(\lambda-i/2)}{Q(\lambda-i/2)}+\frac{q(\lambda+i/2)}{Q(\lambda+i/2)} \\
&+\frac{1}{2}\left(\frac{\sigma(\lambda-i/2)}{\prod_{s\in S}(\lambda-s)(\lambda-i-s)}+\frac{\sigma(\lambda+i/2)}{\prod_{s\in S}(\lambda+i-s)(\lambda-s)} \right) \,.
\end{aligned}
\end{equation}
As any polynomial, $\pi(\lambda)$ can be decomposed as
\begin{equation}
\pi(\lambda)=\rho(\lambda+i/2)-\rho(\lambda-i/2)
\end{equation}
with $\rho(\lambda)$ a polynomial, unique up to an additive constant. Denote now
\begin{equation}
U(\lambda)=\frac{1}{2}\left(\frac{\sigma(\lambda-i/2)}{\prod_{s\in S}(\lambda-s)(\lambda-i-s)}+\frac{\sigma(\lambda+i/2)}{\prod_{s\in S}(\lambda+i-s)(\lambda-s)} \right) \,.
\end{equation}
It can be decomposed as
\begin{equation}
U(\lambda)=\sum_{s\in S} \left( \frac{a_{s}}{\lambda -s}+\frac{b_s^+}{\lambda-(s+i)}+\frac{b_s^-}{\lambda-(s-i)} \right)
\end{equation}
with $a_s,b_s^+,b_s^-$ constants.
Using the property of the digamma function $\psi(x)$,
\begin{equation}
\psi(x+1)-\psi(x)=\frac{1}{x} \,,
\end{equation}
one can rewrite it as
\begin{equation}
U(\lambda)=V(\lambda+i/2)-V(\lambda-i/2) \,,
\end{equation}
where
\begin{equation}
\label{eq:Vdecomp}
V(\lambda)=\sum_{s\in S} \left( -i  (a_s+b_s^++b_s^-)\psi(-i(\lambda-s)+1/2)+\frac{b_s^-}{\lambda -(s-i/2)}-\frac{b_s^+}{\lambda -(s+i/2)} \right) \,.
\end{equation}
Therefore
\begin{equation}
R(\lambda)=\frac{P(\lambda+i/2)}{Q(\lambda+i/2)}-\frac{P(\lambda-i/2)}{Q(\lambda-i/2)}
\end{equation}
with
\begin{equation}
\label{eq:PrhoQ}
P(\lambda)=\rho(\lambda)Q(\lambda)+q(\lambda)+Q(\lambda)V(\lambda) \,.
\end{equation}
Note that since $s\pm i/2$ is a root of $Q(\lambda)$, $P$ is a polynomial if and only if $a_s+b_s^++b_s^-=0$ for all $s\in S$. Recalling \eqref{eq:r0}, one gets
\begin{equation}
\label{eq:pq3}
P(\lambda+i/2)Q(\lambda-i/2)-P(\lambda-i/2)Q(\lambda+i/2)=\lambda^L \,,
\end{equation}
as stated in the theorem. Moreover, $P(\lambda)$ takes the form \eqref{eq:pform} by virtue of \eqref{eq:PrhoQ} and \eqref{eq:Vdecomp}.

To show the second part of the lemma, we proceed as follows. Replacing the $(\lambda\pm i/2)^L$ in the $TQ$ relation \eqref{eq:tq} by relation \eqref{eq:pq3}, one gets
\begin{equation}
\label{eq:tpq}
T(\lambda)=P(\lambda+i)Q(\lambda-i)-P(\lambda-i)Q(\lambda+i)
\end{equation}
Evaluating this relation in $s-i/2$ yields, using the form \eqref{eq:PrhoQ} and the fact that the digamma function has a pole at each nonpositive integer with residue $-1$
\begin{equation}
\label{eq:str2}
T(s-i/2)=P(s+i/2)Q(s-3i/2)+(a_s+b_s^++b_s^-) Q^*(s+i/2)Q(s-3i/2)
\end{equation}
Using now the $TQ$ relation:
\begin{equation}
\label{eq:str1}
T(s-i/2)=\frac{Q^*(s+i/2)}{Q^*(s-i/2)}(s-i)^L+\frac{Q(s-3i/2)}{Q^*(s-i/2)}s^L
\end{equation}
and relation \eqref{eq:pq3} for $\lambda=s+i$,
\begin{equation}
\label{eq:pi/2}
P(s+i/2)=-\frac{(s+i)^L}{Q(s+3i/2)} \,,
\end{equation}
together with the fact that $s=0$ is the only possible string center, one gets from \eqref{eq:str2} and \eqref{eq:str1} that $a_s+b_s^++b_s^-=0$ if and only if
\begin{equation}
\label{eq:genb1}
(-1)^L=\prod_{\lambda_k\neq \pm i/2}\frac{\lambda_k+i/2}{\lambda_k-i/2}\cdot \frac{\lambda_k+3i/2}{\lambda_k-3i/2} \,,
\end{equation}
which concludes the proof.
\end{proof}

\subsection{Polynomiality of $P(\lambda)$ and constructability of the Bethe state}
We remark that \eqref{eq:genb1} is exactly the condition found in \cite{nepomechie,avdeev} for having a physical solution of the Bethe equations. However, \cite{nepomechie,avdeev} only proves that \eqref{eq:genb1} \textit{implies} the physicality of the solution. The purpose of Lemma \ref{th:2} below is to prove the \textit{equivalence} between polynomiality of $P(\lambda)$ and physicality of the solution.

Let us briefly explain our reasoning. In presence of exact strings, the residues in \eqref{unwanted} still vanish as in the case of non-singular Bethe roots; however, the Bethe state $B(\lambda_1) \cdots B(\lambda_n)|0\rangle$ vanishes as well (see \cite{slv,Siddharthan} and Lemma \ref{th:1ter}), and imposing the $TQ$ relation \eqref{eq:tq} alone is then non-conclusive. We want to show that one can find a regularization such that the residues in \eqref{unwanted} vanish faster than $B(\lambda_1^\epsilon) \cdots B(\lambda_n^\epsilon)|0\rangle$ when $\epsilon\to 0$, if and only if $P$ is a polynomial. To that end, we need to understand how fast the Bethe state actually vanishes when $\epsilon\to 0$, which is the purpose of Lemma \ref{th:1ter} (that is in fact needed in \cite{nepomechie} for their reasoning to be conclusive).\\

We should also mention that although the eigenvector constructed within the algebraic Bethe ansatz with the roots of $P$  'beyond the equator' (that has no exact strings) vanishes \cite {faddeevtakh}, it has been observed in the coordinate Bethe ansatz that one can build a non vanishing Bethe vector beyond the equator after some modifications \cite {bax}, and also more recently ideas have been proposed to build eigenstates using the beyond-the-equator Bethe roots \cite{gromovnewconstruction}. These are very elegant ways of building the eigenvector in case of two polynomial solutions to the $TQ$ relations; however, it does not forbid to imagine that the usual Bethe vector could be regularized if $P$ were not a polynomial, whereas Lemma \ref{th:2} does.

\begin{property}
\label{th:1ter}
Let $\lambda_1^\epsilon$ and $\lambda_2^\epsilon$ be such that $\lambda_{1,2}^\epsilon=\pm i/2+O(\epsilon)$ when $\epsilon\to 0$. Then
\begin{equation}
B(\lambda_1^\epsilon)B(\lambda_2^\epsilon)|0\rangle=\begin{cases}
O(\epsilon^L) & \mbox{if } \lambda_1^\epsilon-\lambda_2^\epsilon-i=O(\epsilon^L) \,, \\
O(\lambda_1^\epsilon-\lambda_2^\epsilon-i) & \text{otherwise} \,. \\
\end{cases}
\end{equation}
\end{property}
\begin{proof}
From \eqref{defr} we have (see e.g. \cite{slv})
\begin{equation}
B(\lambda)\rvoid=\sum_{k=1}^L(\lambda+i/2)^{L-k}(\lambda-i/2)^{k-1}i\sigma_k^-\rvoid \,,
\end{equation}
where $\sigma_k^-$ is the Pauli matrix $\left(\begin{smallmatrix}
0&0\\1&0
\end{smallmatrix}\right)$ acting at site $k$. An efficient way to obtain this expression is to apply each $R$ matrix on the quantum space vector $\left(\begin{smallmatrix}1\\0\end{smallmatrix} \right)$ in \eqref{defr} before taking the tensor products. Thus when calculating the monodromy matrix we take a (regular) matrix product of
\begin{equation}
 \left(\begin{matrix}
 \left(\begin{matrix}\lambda+i/2\\0\end{matrix} \right)&\left(\begin{matrix}0\\i\end{matrix} \right)\\
 \left(\begin{matrix}0\\0\end{matrix} \right) & \left(\begin{matrix}\lambda-i/2\\0\end{matrix} \right)
 \end{matrix} \right)
\end{equation}
but whose coefficients have to be tensorized at each site. Similarly
\begin{equation}
\begin{aligned}
B(\lambda)\sigma_k^-\rvoid=&\sum_{q<k} (\lambda+i/2)^{L-q-1}(\lambda-i/2)^{q}i\sigma_q^-\sigma_k^-\rvoid \\
&+\sum_{j>k}(\lambda+i/2)^{L-j+1}(\lambda-i/2)^{j-2}i\sigma_j^-\sigma_k^-\rvoid\\
&-\sum_{q<k<j}(\lambda+i/2)^{L-j+k-q-1}(\lambda-i/2)^{j+q-2-k}i\sigma_q^-\sigma_j^-\rvoid \,.
\end{aligned}
\end{equation}
 Hence
\begin{equation}
\begin{aligned}
\label{2b}
B(\lambda)B(\mu)\rvoid=&\sum_{q<k}\Big(-(\lambda+i/2)^{L-q-1}(\lambda-i/2)^{q}(\mu-i/2)^{k-1}(\mu+i/2)^{L-k} \\
&-(\lambda+i/2)^{L-k+1}(\lambda-i/2)^{k-2}(\mu-i/2)^{q-1}(\mu+i/2)^{L-q}\\
&+\sum_{q<p<k}(\lambda+i/2)^{L-k+p-q-1}(\lambda-i/2)^{k+q-2-p}(\mu-i/2)^{p-1}(\mu+i/2)^{L-p} \Big)\sigma_q^-\sigma_k^-\rvoid \,. \\
\end{aligned}
\end{equation}
By computing the power sum over $p$, after a bit of rearangement one gets
\begin{equation}
B(i/2+\epsilon)B(-i/2+\epsilon)\rvoid=-2\sum_{q<k}\epsilon^{L+k-q-1}(i+\epsilon)^{L-k}(-i+\epsilon)^{q-1}\sigma_q^-\sigma_k^-\rvoid \,,
\end{equation}
which is $O(\epsilon^L)$.\\

For $\lambda=i/2+\epsilon$ and $\mu=-i/2+\epsilon+\eta(\epsilon)$ with $\eta(\epsilon)=O(\epsilon)$, it is clear from \eqref{2b} that this will bring an additional term that is at least $O(\eta)=O(\lambda_1^\epsilon-\lambda_2^\epsilon-i)$, which concludes the proof.
\end{proof}

\begin{property}
\label{th:2}
Let $\{\lambda_1,\ldots,\lambda_n\}$ be a solution to the Bethe ansatz equations. There exists a function $\epsilon\mapsto \{\lambda_1^\epsilon,...,\lambda_n^\epsilon\}$ with $ \underset{\epsilon\to 0}{\lim}\, \lambda_j^\epsilon =\lambda_j$ and $\lambda_k^\epsilon-\lambda_p^\epsilon\neq \pm i$ such that
\begin{equation}
\label{eq:lim}
 \underset{\epsilon\to 0}{\lim}\, \frac{ B(\lambda_1^\epsilon) \cdots B(\lambda_n^\epsilon)|0\rangle}{|| B(\lambda_1^\epsilon) \cdots B(\lambda_n^\epsilon)|0\rangle||}
\end{equation}
exists and is an eigenvector of the transfer matrix, if and only if the function $P(\lambda)$ of Lemma \ref{th:1} is a polynomial.
\end{property}
\begin{proof}

We denote $T_\epsilon(\lambda)$ the function defined by \eqref{eq:tq} in terms of $Q_\epsilon(\lambda)=\prod_j(\lambda-\lambda_j^\epsilon)$:
\begin{equation}
\label{tesp}
T_\epsilon(\lambda)=\frac{Q_\epsilon(\lambda+i)(\lambda-i/2)^L+Q_\epsilon(\lambda-i)(\lambda+i/2)^L}{Q_\epsilon(\lambda)} \,.
\end{equation}
Let us first assume that the regularization is such that $(\lambda_{1,2}^\epsilon\mp i/2)^L=o(\lambda_1^\epsilon-\lambda_2^\epsilon-i)$ for $\lambda_1=i/2$ and $\lambda_2=-i/2$. Then according to Lemma \ref{th:1ter}, we need $\text{Res}(T_\epsilon(\lambda),\lambda_j^\epsilon)$ to be $o(\lambda_1^\epsilon-\lambda_2^\epsilon-i)$. However, from \eqref{tesp} and because of $(\lambda_{1,2}^\epsilon\mp i/2)^L=o(\lambda_1^\epsilon-\lambda_2^\epsilon-i)$ we see that
\begin{equation}
\text{Res}(T_\epsilon(\lambda),\lambda_i^\epsilon)=\frac{i^L Q^*(-i/2)}{Q^*(i/2)}(\lambda_1^\epsilon-\lambda_2^\epsilon-i)+o(\lambda_1^\epsilon-\lambda_2^\epsilon-i) \,,
\end{equation}
which is of the same order as $\lambda_1^\epsilon-\lambda_2^\epsilon-i$ and the residue terms do not vanish faster than the Bethe state. Hence, in any case we need $\lambda_1^\epsilon-\lambda_2^\epsilon-i=O((\lambda_{1,2}^\epsilon\mp i/2)^L)$ to find such a regularization. We will suppose this condition satisfied from now on.\\

We will denote
\begin{equation}
F(\lambda)=\frac{T(\lambda)}{(\lambda+i/2)^L(\lambda-i/2)^L}
\end{equation}
with $F_\epsilon(\lambda)$ its perturbed version, involving $T_\epsilon(\lambda)$.
 Lemma \ref{th:1ter} implies that the condition for the residue terms in \eqref{unwanted} to vanish faster than the Bethe state, in the limit $\epsilon\to 0$, is
\begin{equation}
\label{eq:condition}
 \underset{\epsilon\to 0}{\lim}\, \text{Res}\, (F_\epsilon(\lambda),\lambda_i^\epsilon)=0 \,.
\end{equation}

Let us first build a $P_\epsilon$ corresponding to the $Q_\epsilon$. Decomposing
\begin{equation}
\frac{\lambda^L}{Q_\epsilon(\lambda+i/2)Q_\epsilon(\lambda-i/2)}=\sum_{k}\frac{a_k^+(\epsilon)}{\lambda-(\lambda^\epsilon_k+i/2)}+\frac{a_k^-(\epsilon)}{\lambda-(\lambda^\epsilon_k-i/2)} \,,
\end{equation}
one can write
\begin{equation}
\frac{\lambda^L}{Q_\epsilon(\lambda+i/2)Q_\epsilon(\lambda-i/2)}=U_\epsilon(\lambda+i/2)-U_\epsilon(\lambda-i/2)
\end{equation}
with
\begin{equation}
U_\epsilon(\lambda)=\sum_k -ia_k^+(\epsilon)\psi(-i(\lambda-(\lambda_k^\epsilon+ i/2))+1/2)-ia_k^-(\epsilon)\psi(-i(\lambda-(\lambda_k^\epsilon- i/2))+1/2) \,,
\end{equation}
and so
\begin{equation}
\label{eq:pqeps}
P_\epsilon(\lambda+i/2)Q_\epsilon(\lambda-i/2)-P_\epsilon(\lambda-i/2)Q_\epsilon(\lambda+i/2)=\lambda^L
\end{equation}
with
\begin{equation}
P_\epsilon(\lambda)=Q_\epsilon(\lambda)U_\epsilon(\lambda) \,,
\end{equation}
which has poles at $\lambda_k^\epsilon-ni$ with $n$ a strictly positive integer, with residue $-(a_k^+(\epsilon)+a_k^-(\epsilon))Q_\epsilon(\lambda_k^\epsilon-ni)$.

With relation \eqref{eq:pqeps}, one has
\begin{equation}
F_\epsilon(\lambda)=\frac{P_\epsilon(\lambda+i)Q_\epsilon(\lambda-i)-P_\epsilon(\lambda-i)Q_\epsilon(\lambda+i)}{(\lambda+i/2)^L(\lambda-i/2)^L} \,,
\end{equation}
which has a pole at every $\lambda_k^\epsilon$ with residue
\begin{equation}
r_k(\epsilon)=\frac{(a^+_k(\epsilon)+a^-_k(\epsilon))Q_\epsilon(\lambda_k^\epsilon-i)Q_\epsilon(\lambda_k^\epsilon+i)}{(\lambda_k^\epsilon+i/2)^L(\lambda_k^\epsilon-i/2)^L} \,.
\end{equation}
We now pick a $k$ that corresponds to $i/2$ or $-i/2$, for example without loss of generality $\lambda_k=i/2$. The quantity $(a^+_k(\epsilon)+a^-_k(\epsilon))Q_\epsilon(\lambda_k^\epsilon-i)/(\lambda_k^\epsilon-i/2)^L$ is undetermined when $\epsilon\to 0$. With relation \eqref{eq:pqeps} at $\lambda_k^\epsilon-i/2$, one gets
\begin{equation}
P_\epsilon(\lambda_k^\epsilon)Q_\epsilon(\lambda_k^\epsilon-i)+(a_k^+(\epsilon)+a_k^-(\epsilon))Q_\epsilon(\lambda_k^\epsilon-i)Q_\epsilon^*(\lambda_k^\epsilon)=(\lambda_k^\epsilon-i/2)^L \,,
\end{equation}
whence
\begin{equation}
\label{eq:eq}
\frac{(a^+_k(\epsilon)+a^-_k(\epsilon))Q_\epsilon(\lambda_k^\epsilon-i)}{(\lambda_k^\epsilon-i/2)^L}=\frac{1}{Q^*_\epsilon(\lambda^\epsilon_k)}\left(1-\frac{P_\epsilon(\lambda_k^\epsilon)Q_\epsilon(\lambda^\epsilon_k-i)}{(\lambda^\epsilon_k-i/2)^L}\right) \,.
\end{equation}
The left-hand side vanishes if and only if $a^+_k(\epsilon)+a^-_k(\epsilon)$ vanishes. Indeed, if the left-hand side vanishes, then $\tfrac{Q_\epsilon(\lambda_k^\epsilon-i)}{(\lambda_k^\epsilon-i/2)^L}$ cannot vanish on the right hand-side. If $a^+_k(\epsilon)+a^-_k(\epsilon)$ vanishes, then $\tfrac{Q_\epsilon(\lambda_k^\epsilon-i)}{(\lambda_k^\epsilon-i/2)^L}$ cannot diverge when $\epsilon\to 0$, otherwise the right-hand side would diverge faster since $P(i/2)\neq 0$, see \eqref{eq:pi/2}; and so the whole left-hand side must vanish.

If $P$ is not a polynomial, according to Lemma \ref{th:1} it must have a pole at $-3i/2$, so that $a^+_k(\epsilon)+a^-_k(\epsilon)$ does not vanish when $\epsilon\to 0$, at least for one $k$ such that $\lambda_k=i/2$ or $\lambda_k=-i/2$ (we can assume that it is true for $i/2$; otherwise we could have chosen $-i/2$ before). Hence the left-hand side of \eqref{eq:eq} does not vanish and we cannot have $r_k(\epsilon)\to 0$ when $\epsilon\to 0$.

If $P$ is a polynomial, for an arbitrary function $\epsilon\mapsto \lambda_j^\epsilon$, the different poles $a^+_k(\epsilon)+a^-_k(\epsilon)$ do not necessarily vanish individually in the limit $\epsilon\to 0$, since they can compensate each other (for example, $1/(\lambda-\epsilon)-1/(\lambda+\epsilon)$ does not have any poles in the limit $\epsilon\to 0$, even if the residues at $\epsilon\neq 0$ do not vanish in the limit $\epsilon\to 0$). Coming back to \eqref{eq:pqeps} evaluated at $\lambda=\lambda^\epsilon_k-i/2$ for $\lambda_k=-i/2$ and for  $\lambda_k=i/2$ , one sees that the vanishing of the residues is equivalent to
\begin{equation}
\label{eq:nep}
Q(-3i/2)=\frac{(-i)^L}{P(-i/2)}\,,\qquad Q_\epsilon(\lambda_k^\epsilon-i)\sim \frac{(\lambda_k^\epsilon-i/2)^L}{P(i/2)}\,\quad \text{for }\lambda_k=i/2 \,.
\end{equation}
The first condition is always satisfied when $P$ is a polynomial, and the second one is an additional condition that has to be satisfied for the Bethe vector to be an eigenvector in the limit $\epsilon\to 0$. This shows that if $P$ is a polynomial, then the poles $r_k(\epsilon)$ can vanish in the limit $\epsilon\to 0$ with an appropriate choice of roots $\lambda_k^\epsilon$.

\end{proof}

We remark that the second condition in \eqref{eq:nep}, writing the perturbed roots as $i/2+\epsilon$ and $-i/2+\eta(\epsilon)$, can be translated into
\begin{equation}
\eta(\epsilon)=\epsilon+\frac{\epsilon^L Q(3i/2)}{i^L Q^*(-i/2)}+o(\epsilon^L) \,,
\end{equation}
which was the regularization found in \cite{nepomechie}.

\subsection{An additional $TQ$ relation}

We can now prove the

\begin{theorem}
\label{th:4}
$Q(\lambda)=\prod_{i=1}^n (\lambda-\lambda_i)$ is a physical solution to the Bethe ansatz equations if and only if the functions $T_0(\lambda)$ and $T_1(\lambda)$ in the following two $TQ$ relations are polynomials:
\begin{equation}
\label{eq:t3}
\begin{aligned}
&T_0(\lambda)Q(\lambda)=W_0(\lambda-i/2) Q(\lambda+i)+W_0(\lambda+i/2) Q(\lambda-i) \,, \\
&T_1(\lambda)Q'(\lambda)=W_1(\lambda-i/2) Q'(\lambda+i)+W_1(\lambda+i/2) Q'(\lambda-i) \,,
\end{aligned}
\end{equation}
where
\begin{equation}
\begin{aligned}
Q'(\lambda)=&\ Q(\lambda+i/2)-Q(\lambda-i/2) \,, \\
W_0(\lambda)=& \ \lambda^L \,, \\
W_{1}(\lambda)=& \ W_0(\lambda+i/2)+W_0(\lambda-i/2)-T_0(\lambda) \,.
\end{aligned}
\end{equation}
\end{theorem}
\begin{proof}
It is straightforward to show that
\begin{equation}
\label{eq:w}
W_1(\lambda)=Q'(\lambda-i/2)P'(\lambda+i/2)-Q'(\lambda+i/2)P'(\lambda-i/2)
\end{equation}
where $P'(\lambda)=P(\lambda+i/2)-P(\lambda-i/2)$ with $P(\lambda)$ the function introduced in Lemma \ref{th:1}, using equation \eqref{eq:pq}. Then
\begin{equation}
T_1(\lambda)=P'(\lambda+i)Q'(\lambda-i)-P'(\lambda-i)Q'(\lambda+i) \,.
\end{equation}
Now, from the general form of $P$ in Lemma \ref{th:1}, one has
\begin{equation}
P'(\lambda)=A(\lambda)+Q'(\lambda)\alpha_0 \psi (-i\lambda)
\end{equation}
with $A(\lambda)$ a rational function with a unique simple pole at $0$ with residue proportional to $\alpha_0$, using $\psi(x+1)-\psi(x)=1/x$. Since $\psi$ has a pole at $-1$, $T_1$ has a priori a pole at $0$ with residue $i\alpha_0 Q'(-i)Q'(i)$. From Lemma \ref{th:0}, $\pm i$ are never center of strings and so $Q'(\pm i)=\pm Q(\pm 3i/2)\neq 0$ if $\alpha_0\neq 0$. It follows that $T_1$ is a polynomial if and only if $P$ is. Then Lemma \ref{th:2} concludes the proof.
\end{proof}

\subsection{The algorithm of Marboe and Volin}

Let us now come back to the algorithm of Marboe and Volin \cite{marboe}. It consists in introducing functions $Q_{a,s}$ with $s=0,\ldots,L-K$, for $a=0,1,2$ if $s\leq K$ and $a=0,1$ if $s>K$, satisfying the following $QQ$ relations
\begin{equation}
\label{eq:QQrel}
Q_{a+1,s}(\lambda)Q_{a,s+1}(\lambda)\propto Q_{a+1,s+1}(\lambda+i/2)Q_{a,s}(\lambda-i/2)- Q_{a+1,s+1}(\lambda-i/2)Q_{a,s}(\lambda+i/2)
\end{equation}
with the boundary conditions $Q_{0,0}(\lambda)=\lambda^L$, $Q_{2,s}=1$ for $s\leq K$, $Q_{1,s}=1$ for $s>K$, and imposing that all the $Q_{a,s}$ are polynomials. The $Q(\lambda)$ is then given by $Q_{1,0}(\lambda)$.

The labels $(a,s)$ can be interpreted as the coordinates of corners of boxes in an associated two-row Young diagram
\begin{equation}
\begin{tikzpicture}
 \draw (0,0)--(8,0);
 \draw (0,1)--(8,1);
 \draw (0,2)--(4,2);
 \draw (0,0)--(0,2);
 \draw (1,0)--(1,2);
 \draw (2,0)--(2,2);
 \draw (3,0)--(3,2);
 \draw (4,0)--(4,2);
 \draw (5,0)--(5,1);
 \draw (6,0)--(6,1);
 \draw (7,0)--(7,1);
 \draw (8,0)--(8,1);
 \draw (0,0) node[below] {0};
 \draw (1,0) node[below] {1};
 \draw (4,0) node[below] {$K$};
 \draw (8,0) node[below] {$L-K$};
 \draw (0,0) node[left] {0};
 \draw (0,1) node[left] {1};
 \draw (0,2) node[left] {2};
 \tikzset{>=latex}
 \draw[->] (-0.7,0)--(-0.7,1);
 \draw (-0.7,0.5) node[left] {a};
 \draw[->] (0,-0.7)--(1,-0.7);
 \draw (0.5,-0.7) node[below] {s};
 \draw (2.5,0.5) node {$\cdots$};
 \draw (2.5,1.5) node {$\cdots$};
 \draw (6.5,0.5) node {$\cdots$};
 \end{tikzpicture} \nonumber
\end{equation}
Each $QQ$ relation \eqref{eq:QQrel} then imposes a constraint on the four $Q$-functions associated with the corners of the box whose lower left corner is $(a,s)$. The boundary conditions fix in particular $Q_{a,s} = 1$ for all corners along the top of the diagram.

In this context, it is readily checked that the two $TQ$ relations \eqref{eq:t3} are exactly the relations obtained when $Q_{a,s}$ are imposed to be polynomials for $a=0,1,2$ and $s=0,1,2$. In other words, they are the zero-remainder conditions associated with the two boxes in the leftmost column of the Young diagram. (Note that in the special case $K=1$ there cannot be strings and the second equation of \eqref{eq:t3} is trivially satisfied.) Thus, according to Theorem \ref{th:4}, all the other polynomials $Q_{a,s}$ for $s>2$, as well as the corresponding relations \eqref{eq:QQrel} fixing them, are actually superfluous, as conjectured in \cite{marboe}.

We also remark that the fact that only one additional $TQ$ relation is needed to discard the unphysical solutions is linked to the fact that there is only one possible exact string (otherwise this $TQ$ relation would only give one equation relating the $\alpha_s$'s).\\

\subsection{An example}
Let us illustrate Theorem \ref{th:4} with sizes $L=4$ and $L=5$. In both cases the polynomial $Q(\lambda)=(\lambda+i/2)(\lambda-i/2)=\lambda^2+\tfrac{1}{4}$ is a solution to the first $TQ$ relation with
\begin{equation}
T_0(\lambda) = \begin{cases}
-\tfrac{3}{8}+3\lambda^2+2\lambda^4\,, & \text{if }L=4 \,, \\
-\tfrac{11}{8}\lambda+3\lambda^3+2\lambda^5\,, & \text{if }L=5 \,.
\end{cases}
\end{equation}
However, the corresponding $T_1(\lambda)$ reads
\begin{equation}
T_1(\lambda) = \begin{cases}
-4(2+3\lambda^2)\,,& \text{if }L=4 \,,\\
 \frac{4}{\lambda}-8\lambda-16\lambda^3\,, & \text{if }L=5 \,,
\end{cases}
\end{equation}
showing that the polynomiality of the solution to the second $TQ$ relation is satisfied for $L=4$, but not for $L=5$. Besides, the function $P(\lambda)$ of Lemma \ref{th:1} reads
\begin{equation}
P(\lambda) = \begin{cases}
-i\lambda\left(\lambda^2+\tfrac{5}{4}\right)\,, & \text{if }L=4 \,, \\
\frac{1}{2i}\lambda^2 \left( \lambda^2+\tfrac{1}{4}\right)+\tfrac{i}{2}+\left(\lambda^2+\tfrac{1}{4}\right)i\psi(-i\lambda+1/2)\,, & \text{if }L=5
\end{cases}
\end{equation}
and is a polynomial if and only if $T_1(\lambda)$ is a polynomial. It turns out that $T_0(\lambda)$ is indeed an eigenvalue of the transfer matrix for $L=4$, but not for $L=5$, in agreement with the theorem.

In Figure \ref{fig:roots} we plot the roots of all the polynomials $Q(\lambda)$ solution to the $TQ$ relation \eqref{eq:tq} in size $L=6$, showing in blue those whose solve the two $TQ$ relations \eqref{eq:t3} and in red those that only solve the first one. Only the solutions $(-i/2,i/2)$ $(-i/2,0,i/2)$ among the blue ones involve exact strings in Figure \ref{fig:roots} (all the red non-physical solutions must exhibit exact strings). 

We see in this example that the number of admissible solutions with $K$ roots is ${L\choose K}-{L\choose K-1}$. We recall that in the Heisenberg spin chain the Bethe states are necessarily highest-weight states with respect to the underlying $su(2)$ algebra. Taking into account that the eigenvalue corresponding to a solution with $K$ Bethe roots is $(L-2K+1)$-fold degenerate, one obtains $2^6$ eigenstates indeed.

\begin{figure}[H]
 \begin{center}
\includegraphics[width=0.24\columnwidth]{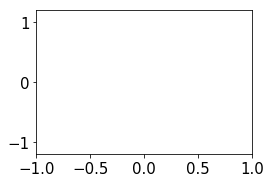} 
\includegraphics[width=0.24\columnwidth]{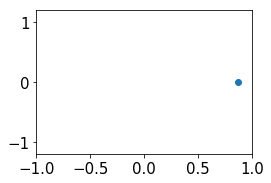} 
\includegraphics[width=0.24\columnwidth]{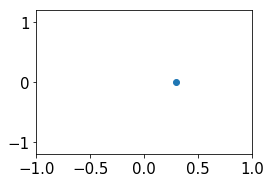} 
\includegraphics[width=0.24\columnwidth]{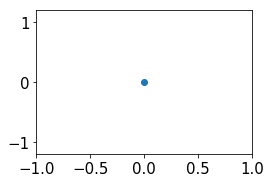} 
\includegraphics[width=0.24\columnwidth]{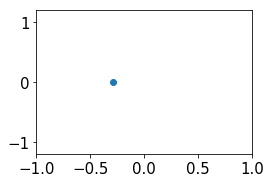}
\includegraphics[width=0.24\columnwidth]{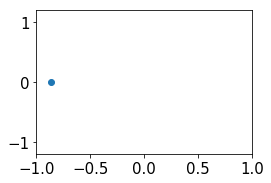}
\includegraphics[width=0.24\columnwidth]{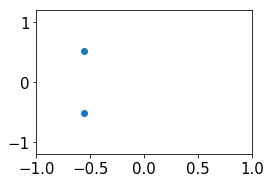}
\includegraphics[width=0.24\columnwidth]{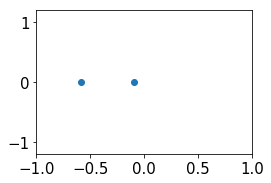}
\includegraphics[width=0.24\columnwidth]{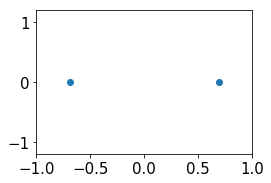}
\includegraphics[width=0.24\columnwidth]{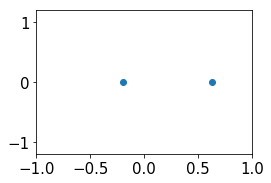}
\includegraphics[width=0.24\columnwidth]{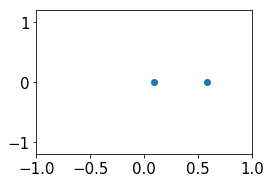}
\includegraphics[width=0.24\columnwidth]{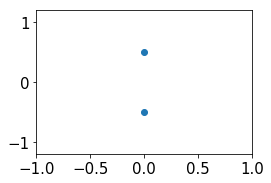}
\includegraphics[width=0.24\columnwidth]{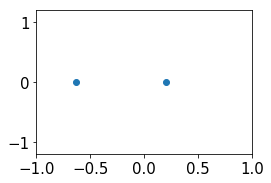}
\includegraphics[width=0.24\columnwidth]{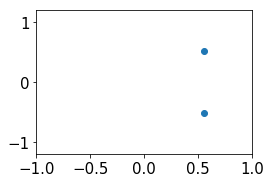}
\includegraphics[width=0.24\columnwidth]{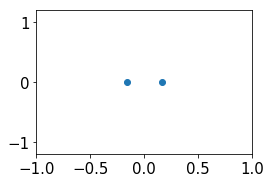}
\includegraphics[width=0.24\columnwidth]{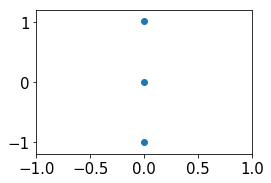}
\includegraphics[width=0.24\columnwidth]{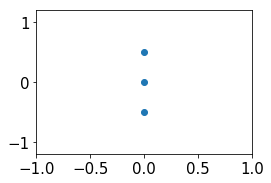}
\includegraphics[width=0.24\columnwidth]{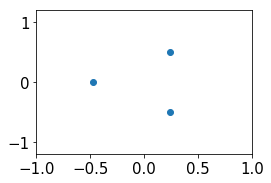}
\includegraphics[width=0.24\columnwidth]{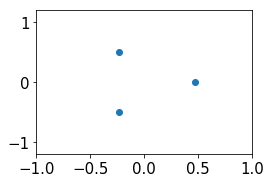}
\includegraphics[width=0.24\columnwidth]{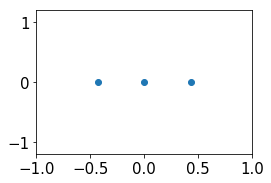}
\includegraphics[width=0.24\columnwidth]{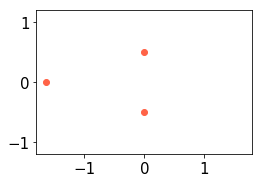}
\includegraphics[width=0.24\columnwidth]{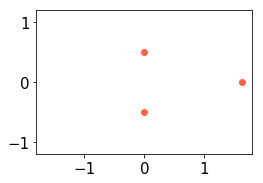}
\includegraphics[width=0.24\columnwidth]{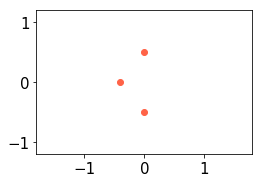}
\includegraphics[width=0.24\columnwidth]{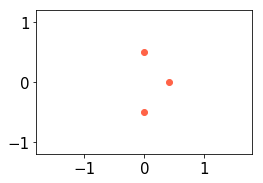}
\end{center}
\caption{In blue: the roots of all the solutions $Q(\lambda)$ to the two $TQ$ relations \eqref{eq:t3} in size $L=6$. In red: the roots of the solutions to the first $TQ$ relation in \eqref{eq:t3} that are not solution to the second one, and thus that do not contribute to the spectrum.}
 \label{fig:roots}
\end{figure}

\subsection*{Acknowledgements}
We are very grateful to R.I.~Nepomechie for his careful reading of the paper.
This work was supported by the ERC Advanced Grant NuQFT and by the EPSRC under grant EP/S020527/1.

\bibliography{expansion}

\bibliographystyle{ieeetr}

\end{document}